\documentclass[11pt]{article}
\usepackage{amssymb, latexsym, mathtools, amsthm}
\begingroup
    \makeatletter
    \@for\theoremstyle:=definition,remark,plain\do{%
        \expandafter\g@addto@macro\csname th@\theoremstyle\endcsname{%
            \addtolength\thm@preskip\parskip
            }%
        }
\endgroup
\usepackage{enumerate}
\usepackage{pgf,tikz}
\usepackage{pdfsync}
\usepackage{txfonts}
\usepackage[T1]{fontenc}
\usepackage{booktabs}
 \usepackage{graphicx}
\definecolor{dnrbl}{rgb}{0,0,0.3}
\definecolor{dnrgr}{rgb}{0,0.3,0}
\definecolor{dnrre}{rgb}{0.5,0,0}
\usepackage[colorlinks=true, citecolor=dnrgr, linkcolor=dnrre, urlcolor=dnrre] {hyperref}
\usepackage{xcolor}
\usepackage{parskip}
\usepackage{tabularx, colortbl}

\usepackage{geometry}
 \usepackage[scriptsize, up]{caption}
\usepackage{pgf,tikz}
\usetikzlibrary{decorations, decorations.pathmorphing}
\usetikzlibrary {shapes}

\theoremstyle{plain}
\newtheorem{thm}{Theorem}[section]
\newtheorem{prop}[thm]{Proposition}
\newtheorem{lem}[thm]{Lemma}

\newtheorem{defi}[thm]{Definition}
\numberwithin{equation}{subsection}
\makeatletter

\let\c@table\c@figure
\makeatother

\newcommand{\Nat}{\mathbb{N}}

\newcommand{\restr}{\upharpoonright}  
\DeclarePairedDelimiter{\ceil}{\lceil}{\rceil}

\newcommand{\bigo}[1]{\mathop{\bf O}\/\left({#1}\right)}

\newcommand{\FSW}{Figueira, Stephan, and Wu\ }

\newcommand{\KS}{Ku{\v{c}}era and Slaman\ }
\newcommand{\CHKW}{Calude, Hertling, Khoussainov and Wang\ }

\newcommand{\DHMN}{Downey, Hirschfeldt Miller and Nies\ }
\newcommand{\ml}{Martin-L\"{o}f }
\newcommand{\pz}{$\Pi^0_1$\ }

\newcommand{\eg}{e.g.\ }
\newcommand{\ie}{i.e.\ }
\newcommand{\ce}{c.e.\ }

\newcommand{\lce}{left-c.e.\ }

\newcommand{\lcern}{left-c.e.\ real}

\newcommand{\lcen}{left-c.e.}

\newcommand{\pf}{prefix-free }

\renewenvironment{abstract}
 { \normalsize
  \list{}{
    \setlength{\leftmargin}{.0cm}%
    \setlength{\rightmargin}{\leftmargin}%
    }%
  \item {\bf \abstractname.} \relax}
 {\endlist}

\title{Computing halting probabilities from other halting probabilities
\thanks{Barmpalias was supported by the 
1000 Talents Program for Young Scholars from the Chinese Government, grant no.\ D1101130.
Additional support was received by
the Chinese Academy of Sciences (CAS) and the Institute of Software of the CAS.
Lewis-Pye was supported by a Royal Society University 
Research Fellowship.}}
\author{George Barmpalias  \and Andrew Lewis-Pye}
\date{\today}
\begin{document}
\maketitle
\begin{abstract}
The halting probability of a Turing machine is the probability that the machine will halt if it starts with a random
stream written on its one-way input tape. When the machine is universal, this probability is referred to as
Chaitin's omega number, and is the most well known example of a real which is random in the sense of Martin-L\"{o}f. Although omega numbers depend on the underlying universal Turing machine,
they are robust in the sense that they all have the same Turing degree,
namely the degree of the halting problem.
This means that, given two universal prefix-free machines $U,V$, 
the halting probability $\Omega_U$
of $U$ computes
the halting probability  $\Omega_V$ of $V$. 
If this computation uses at most the first $n+g(n)$ bits of $\Omega_U$ for the computation of the first $n$ bits of
$\Omega_V$, we say that
$\Omega_U$
computes
$\Omega_V$ with redundancy $g$.

In this paper we  give precise bounds on the redundancy growth rate that is 
generally required for the computation of an omega number from another omega number.
We show that for each $\epsilon>1$, any pair of omega numbers compute each other with redundancy 
$\epsilon\log n$. On the other hand, this is not true for $\epsilon=1$. In fact, we show
that for each omega number $\Omega_U$
there exists another omega number which is not computable from 
$\Omega_U$ with redundancy $\log n$.
This latter result improves an older result of Frank Stephan.
\end{abstract}
\vfill
\noindent{\bf George Barmpalias}\\[0.5em]
\noindent State Key Lab of Computer Science, 
Institute of Software, Chinese Academy of Sciences, Beijing, China.
School of Mathematics, Statistics and Operations Research,
Victoria University of Wellington, New Zealand.\\[0.2em] \textit{E-mail:} \texttt{barmpalias@gmail.com}.
\textit{Web:} \texttt{\href{http://barmpalias.net}{http://barmpalias.net}}
\vfill
\noindent{\bf Andrew Lewis-Pye}\\[0.5em]
\noindent Department of Mathematics,
Columbia House, London School of Economics, 
Houghton Street, London, WC2A 2AE, United Kingdom.\\[0.2em]
\textit{E-mail:} \texttt{\textcolor{dnrgr}{A.Lewis7@lse.ac.uk.}}
\textit{Web:} \texttt{\textcolor{dnrre}{http://aemlewis.co.uk}} 
 \vfill\thispagestyle{empty}
\clearpage


\section{Introduction}
Consider the following experiment, involving a Turing machine with a one-way
input tape. We turn on the machine, and whenever it tries to read the next bit from the input,
we give to it a random digit. What is the probability that the machine will halt at some point?
This is an experiment that Chaitin considered in \cite{MR0411829}. In the case of a universal
machine he called the probability $\Omega$ and showed that it is algorithmically random, in the 
sense of \ml \cite{MR0223179}.
Chaitin originally considered $\Omega$ for self-delimiting machines, \ie machines that
operate on instantaneous code, without any
out-of-band markers or special symbols that frame the words in the input tape.
The cumulative work of Solovay \cite{Solovay:75},
\CHKW \cite{Calude.Hertling.ea:01} and \KS \cite{Kucera.Slaman:01},
has  shown Chaitin's omega numbers  do not depend
significantly on many parameters of the universal machine.
In particular, the series of papers above showed that
a \lce real (\ie one which is the limit of a computable increasing sequence of rational numbers)
is the halting probability of a universal self-delimiting machine if and only if it is
\ml random (a good presentation of this work is given 
in \cite[Chapter 9]{rodenisbook}).

There are many other results that witness the robustness of the halting probability
and the similarity between different omega numbers. Solovay \cite{Solovay:75}, for example, 
showed that omega numbers are, in a specific sense, equally and maximally hard to approximate,
compared to other \lce reals.  Calude and Nies observed in \cite{Calude.Nies:nd} that
omega numbers are computable from each other, with computable bounds on the use of the oracle (i.e.\ computable bounds on the number of bits of the oracle tape required on each argument).
On the other hand, a number of incompatibility results are known which distinguish the halting probabilities
of different machines. 
\FSW \cite{jc/FigueiraSW06}
showed, for example,  that for each universal machine $U$ with halting probability $\Omega_{U}$
there exists a universal machine $V$
with halting probability $\Omega_{V}$
such that $\Omega_{U}$ and $\Omega_{V}$ 
have incomparable truth-table degrees.
 Frank Stephan
(see \cite[Section 6]{IOPORT.05678491} for a proof) showed that
for each universal machine $U$ there exists a universal machine $V$
such that $\Omega_{U}$ cannot
compute the first $n$ bits of $\Omega_{V}$ using only the first $n+\bigo{1}$ bits of
$\Omega_{U}$ as an oracle.
Tadaki \cite{Tadaki:2009:C9} gave a very interesting quantitative characterization of the equivalence between the initial segments of
$\Omega$ and the sets $A_n$ of the strings of length $n$ in the domain of the universal \pf machine.

\DHMN showed  \cite{MR2188515} that the Turing degree of $\Omega$ is not robust 
when the halting probability is relativized to an infinite oracle, even when two oracles differ
at only a finite number of bits. Such strong negative results do not only apply to relativized versions
of halting probabilities, but also to probabilities that concern more complicated properties of
a universal machine than mere halting. This was demonstrated by Barmpalias and Dowe in
\cite{Barmpalias3488}, who studied the probability that a machine 
remains universal even when random bits are prefixed in the input tape. This is known as the
{\em universality probability}, 
and it was shown that for different universal Turing machines the universality
probabilities can have different Turing degrees.
These negative results suggest that the apparent robustness of $\Omega$ stems from the
fact that it is the probability of a relatively simple property, namely halting, which is $\Sigma^0_1$.
Indeed, there is only one $\Sigma^0_1$ Turing degree which contains \ml random numbers, namely
the degree of the halting problem, but the same is not true for classes of higher arithmetical complexity.

In this paper we study the similarity of omega numbers in terms of the length of the initial segment of
an omega number $\Omega_0$ 
that is needed in order to compute the first $n$ bits of another omega number $\Omega_1$.
\begin{defi}[Redundancy]
If a real $\beta$ computes a real $\alpha$, and for each $n$ the computation of the first $n$ bits of
$\alpha$ uses at most the first $\lfloor n+g(n) \rfloor $ bits of the oracle $\beta$, we say that $\beta$
computes $\alpha$ with redundancy $g$.
\end{defi}

Our main result is a sharp estimation of the redundancy growth rate that is generally required
for the computation of one omega number from another, in terms of logarithms. Throughout this paper we shall write $\mbox{log} (n)$ in order to denote $ \mbox{log}_2 (n) $, i.e.\ we always work base 2. It will also be convenient to agree to the convention that $\mbox{log} (0)=0$.

\begin{defi} \label{funcs}
For $\epsilon\in \mathbb{R}$ with $\epsilon\geq 1$, we define $h_{\epsilon}(n)=\epsilon\cdot\log (n)$, and  $h^{\ast}_{\epsilon}(n)=\log (n)+\epsilon\cdot \log \log (n)$. 
\end{defi} 

\begin{thm}\label{uVwUvgKQIJ}
 If $\epsilon>1$ then
every omega number is computable from any other omega number with redundancy $h_{\epsilon}$.
If $\epsilon=1$ then given any omega number $\Omega$ there exists another omega number which is
not computable from $\Omega$ with redundancy $h_{\epsilon}$.
\end{thm}

 Our result extends an older result of Frank Stephan
(see \cite[Section 6]{IOPORT.05678491} for a proof) which says that two omega numbers do not always compute each other with constant redundancy.
Our proof of Theorem \ref{uVwUvgKQIJ} involves effective measure-theoretic tools like
the effective Borel-Cantelli lemmas, effective martingales and other notions from algorithmic randomness.
We review these notions in Section \ref{D4zxFYtrR5} and present our main argument in 
Section \ref{3QK4gQEuwa}.

\section{Preliminaries}\label{D4zxFYtrR5}
Recall that the series $\sum_n n^{-{\epsilon}}$ converges if and only if $\epsilon>1$.
Recall also the Cauchy condensation series convergence
criterion (\eg see  \cite{conde}).
\begin{lem}[Condensation]\label{eq:condtest}
If $f:\mathbb{N} \rightarrow \mathbb{R}^{+}$ is  nonincreasing,  then the series
$\sum_n f(n)$ converges if and only if
the series $\sum_n \big(2^n\cdot f(2^n)\big)$ converges.
\end{lem}
From these well known facts the following lemma immediately follows.  
\begin{lem}[Convergence and divergence] \label{QaQK5kYLLF} 
The sums
 $\sum_n 2^{-h_{\epsilon}(n)}$  and  $\sum_n 2^{-h^{\ast}_{\epsilon}(n)}$ are finite if and only if $\epsilon>1$. 
 \end{lem}

Background on algorithmic randomness that is relevant to our argument can be found in
\cite[Chapter 6]{rodenisbook}. This monograph also contains a presentation of
the work in \cite{Solovay:75}.  We shall identify reals with their infinite binary expansions (the fact that dyadic rationals have two expansions will not cause issues). We shall generally work with reals in $[0,1]$, so that the decimal point may be neglected and reals can be thought of simply as infinite binary sequences, i.e.\ elements of Cantor space. It will be convenient to adopt the (slightly unusual) convention that the bits of a real $\alpha$ are indexed from 1 rather than zero, so that $\alpha = \alpha(1)\alpha(2)\alpha(3)\cdots$, rather than $\alpha(0) \alpha(1)\alpha(2) \cdots$.

A real is \ml random  if it avoids all effective statistical tests.
This notion was introduced by \ml in \cite{MR0223179}.
We will make use of an essentially equivalent notion of statistical test due to Solovay \cite{Solovay:75}:  a Solovay test is a
computable sequence of finite strings $(\sigma_i)$ (each $\sigma_i$ often being identified with the set of infinite binary sequences extending it, meaning that it may be regarded as a basic open subset of Cantor space)
such that $\sum_i 2^{-|\sigma_i|}$ is bounded. 
We say that a real avoids this test if there are only finitely many $i$ such that $\sigma_i$ is
a prefix of the binary expansion of the real. Solovay showed that a real is \ml random if and only if it avoids all Solovay tests.
An equivalent definition of \ml randomness can be given in terms of 
betting strategies, which are often expressed as martingales.
We shall think of  martingales as functions 
$f: 2^{<\omega}\to\mathbb{R}^{\geq 0}$
with the property $f(\sigma 0)+f(\sigma 1)= 2\cdot f(\sigma)$. A martingale $f$ 
is computably enumerable (c.e.)
if the values  $f(\sigma)$ can be computably and uniformly approximated by rationals from below, i.e.\ there exists a computable function $f^{\ast}(\sigma,i)$ taking rational values, which is nondecreasing in the second argument, and such that for all $\sigma$, $\mbox{lim}_{i\rightarrow \infty} f^{\ast}(\sigma,i)=f(\sigma)$. 
We say that $f$ succeeds on a real $X$ if $\lim_n f(\alpha\restr_n)=\infty$.
It is well known that a real $X$ is \ml random if and only if no \ce martingale succeeds on $X$.
We let $\alpha \restr_n$ denote the first $n$ bits of $\alpha$. Given a real $\alpha$, suppose that there exists  a partial computable function $p$ such that 
 $p(\alpha\restr_n)\downarrow $ for infinitely many $n$, and such that 
whenever $p(\alpha\restr_n)\downarrow $ we have $p(\alpha\restr_n)=\alpha(n+1)$, i.e.\ $p$ correctly predicts the next bit of $\alpha$ (recall our labelling convention above). In this case it is not hard to see that there exists a c.e.\ martingale which succeeds on $\alpha$, so that $\alpha$ cannot be Martin-L\"{o}f random. This fact will be used in the proof of Theorem
\ref{uVwUvgKQIJ}.  A real is weakly 1-random if it is not a member of any  null \pz class.

Finally, we state the effective Borell-Cantelli lemmas that are often used in order to derive statistical properties
of algorithmically random numbers.
Recall the basic fact from analysis that,
given a sequence $(b_i)$ of positive integers, we have: 
\begin{equation}\label{6m2fJHPIDt}
\prod_i (1-2^{-b_i})>0\iff \sum_i 2^{-b_i}<\infty.
\end{equation}
This is a direct consequence of the fact that
$\log(1+x)=x+\bigo{x^2}$ in a neighborhood of zero.

Given a finite set $B$ of natural numbers, a string $\sigma$ of length $|B|$ (\ie the cardinality of $B$) and a real $\beta$, 
we may say that   $\beta$ {\em meets}
 $\sigma$ {\em on} $B$  if the following holds for all $n< |B|$: 
if $m_n$ is the $n$th element of $B$
 we have $\beta(m_n)=\sigma (n)$.
 The same definition applies for the case when $\beta$ is a string of length at least the largest element of $B$.
Note that if the $B_i$ are disjoint sets (and fixing the uniform probability measure), the events `$\beta$ meets $\sigma_i$ on $B_i$' are independent.
We can therefore state the effective Borel-Cantelli lemmas as follows.
\begin{lem}[Effective Borel-Cantelli lemmas]\label{le:combina2}
Let 
 $(B_i)$ be a collection of pairwise
disjoint sets, let $b_i=|B_i|$ 
for each $i$ and
suppose that $(\sigma_i)$ is a
sequence of strings  
with $|\sigma_i|=|B_i|$ .
\begin{enumerate}
\item If $\sum_i 2^{-b_i}<\infty$ then the reals that meet $\sigma_i$ on $B_i$ for only finitely many $i$ 
form a class of measure 1.
\item If $\sum_i 2^{-b_i}=\infty$ then the reals that meet $\sigma_i$ on $B_i$  for only finitely many $i$ 
form a class of measure 0.
\end{enumerate}
Suppose the sequence  $(B_i)$ is computable. Then in the first case every \ml random real  
meets $\sigma_i$ on $B_i$  for only finitely many $i$,
and in the second case  every weakly 1-random real  meets
$\sigma_i$ on $B_i$ for infinitely many $i$. 
\end{lem}

\noindent The first clause is essentially just Solovay's characterization of \ml randomness in terms of Solovay tests that we discussed above.  
For the sake of completeness we include a short proof of the second clause. 
Let $m_n=\max \cup_{i\leq n} B_i$.
For a given $n$,
the number of subsets of $\{0, \dots, m_n\}$, regarded as strings of length $m_n+1$, which do not meet
$\sigma_i$ on $B_i$ for any $i\leq n$ is:
\begin{equation}\label{eq:combh2hd}
2^{m_n+1-\sum_{i\leq n} b_i} \cdot\prod_{i=0}^{n} (2^{b_i}-1).  
\end{equation}
Since there are $2^{m_n+1}$ subsets of $\{0,\dots, m_n\}$, the measure of reals that
do not meet any of the sets $B_i, i\leq n$ is exactly
the expression in \eqref{eq:combh2hd} divided by $2^{m_n+1}$, \ie
\[
2^{-\sum_{i\leq n} b_i} \cdot\prod_{i=0}^{n} (2^{b_i}-1)=
\prod_{i=0}^{n} (1-2^{-b_i}). 
\]
By \eqref{6m2fJHPIDt}, the above quantity tends to zero if and only if the sum
$\sum_i 2^{-b_i}$
diverges. For each finite set $D$, the sum $\sum_{i\notin D}2^{-b_i}$ still diverges. So the argument above suffices to show that the reals which meet some $B_i$ with $i\notin D$ are of measure 1. Taking the intersection over all finite D, we get a countable intersection of sets of measure 1, which is therefore of measure 1. 
The effective version of the second clause follows since for each $n$, 
the set of reals which meet $\sigma_i$ on $B_i$ for at most $n$ many $i$, forms a null \pz class.

The Borel-Cantelli lemmas were used by Chaitin in \cite{chaitinincomp} in order to establish the 
existence of certain blocks of zeros in the binary expansion of  $\Omega$.
For example, it was shown that if $g$ is computable and $\sum_n 2^{-g(n)}$ diverges,
then for infinitely many $n$ there exists a block of $n+g(n)$ zeros between digits
$2^n$ and $2^{n+1}$ of the binary expansion of $\Omega$.

\section{Upper bounds on the oracle use in computations from omega numbers}
We prove Theorem \ref{uVwUvgKQIJ}, along with some slightly more 
general statements.
In Section \ref{B1WbTtcepf} we consider the more general question of
characterising the computable functions that are upper bounds on the oracle use
in computations  of one halting probability from another one.

\subsection{Proof of Theorem \ref{uVwUvgKQIJ}}\label{3QK4gQEuwa}
We start with the first part of Therorem \ref{uVwUvgKQIJ}, which relies on the
approximation properties of omega numbers.
The limits of increasing 
computable sequences of rational numbers are known as left computably
enumerable (\lcen) reals, and can be viewed as the halting probabilities of
(not necessarily universal) \pf machines. 

\begin{lem}[Sufficient redundancy]\label{F4dKJufHui}
Suppose that $g$ is a computable function such that
$\sum_i 2^{-g(i)}$ converges.
Given any two omega numbers,  
each is computable from the other with redundancy $g$.
\end{lem}
\begin{proof}
Let $g$ be as in the statement, let $\Omega$ be an omega number and let 
$\alpha$ be a \lcern. It suffices to show that $\alpha$ is computable from $\Omega$
with redundancy $g$. Let $(\alpha_s), (\Omega_s)$ be computable nondecreasing 
dyadic rational approximations that converge to $\alpha,\Omega$ respectively.
Recall that a Solovay test is a computable sequence of basic open intervals $(\sigma_i)$
such that $\sum_i 2^{-|\sigma_i|}$ is bounded above.
Since $\Omega$ is \ml random, it has only finitely many initial segments in any  Solovay test $(\sigma_i)$. 
We construct a Solovay test as follows, along with a \ce set $I$. 
At each stage $s+1$ we consider the least
$n\leq s$ such that $\alpha_s(n)\neq \alpha_{s+1}(n)$, if such exists. If such an $n$ exists,  we define
$\sigma_{s}=\Omega_{s+1}\restr_{\lfloor n+g(n)\rfloor }$ and enumerate $s$ into  $I$.
First let us verify that the set of strings $\sigma_s, s\in I$ is a Solovay test. Note that for every $n$,
the number of stages $s$ such that $n$ is the least number with the property that
 $\alpha_s(n)\neq \alpha_{s+1}(n)$, is bounded above by 
 the number of times that $\alpha_s(n)$ can change from
0 to 1 in the monotone approximation to $\alpha$. Hence this number is bounded above
by $2^{n-1}$. So we have: 
\[
\sum_{s\in I} 2^{-|\sigma_s|}\leq 
\sum_n 2^n\cdot 2^{-g(n)-n}=\sum_n 2^{-g(n)}<\infty.
\]
Since $\Omega$ is \ml random, there exists some $s_0$ such that for $s>s_0$ in $I$, $\sigma_s$ is not 
an initial segment of  $\Omega$. This means that whenever our construction enumerates
$s$ in $I$ because we find some least  $n$ with $\alpha_s(n)\neq \alpha_{s+1}(n)$, 
there exists some later stage where the approximation to $\Omega\restr_{\lfloor n+g(n)\rfloor}$ changes. So with oracle $\Omega\restr_{s+g(s)}$ we can uniformly compute $\alpha(n)$, and 
$\alpha$ is computable from $\Omega$ with redundancy $g$.
\end{proof}
The reader may note that the above proof establishes a slightly more general statement, regarding
the computation of any \lce real from an omega number.
The second part of Theorem \ref{uVwUvgKQIJ} is also established slightly more generally than stated,
as the following lemma indicates. For this proof recall that, by
Demuth \cite{Demuth:75ce},
the sum of a \ml random \lce real and any other \lce real is \ml random.
Since the halting probabilities of universal \pf machines are exactly the \ml random \lce reals, it follows 
that the sum of an omega number and any \lce real is an omega number.


\begin{lem}[Insufficient redundancy]\label{23akY6Ldyu}
Let $g$ be a computable nondecreasing function and let $(t_i)$
be a computable increasing sequence such that $t_i+g(t_i)<t_{i+1}$ for all sufficiently large $i$ and: 
\begin{equation}\label{SJlX3zHJot}
\sum_i 2^{-g(t_i)}=\infty
\hspace{0.7cm}\textrm{and}\hspace{0.7cm}
\sum_i 2^{t_i-t_{i+1}}<\infty.
\end{equation}
Then given any omega number $\Omega$ there exists another omega number which is
not computable from $\Omega$ with redundancy $g$.
\end{lem}
\begin{proof}
We will show that for some constant $c$ the following number has the required property:  
\begin{equation}\label{HfHfPGgk7}
\beta =\Omega +\sum_{i>c} 2^{-(t_i+\lfloor g(t_i) \rfloor+1)}.
\end{equation}
First, note that $\beta$ as defined above is an omega number, since it is 
the sum of an omega number and a computable real.
Consider the intervals of  positions $I_k=[t_k,t_k+\lfloor g(t_k) \rfloor]$ and 
$J_k=[t_k+\lfloor g(t_k) \rfloor+2, t_{k+1}+ \lfloor g(t_k) \rfloor+1]$.
Given a real number $X$, we are interested in 
those  $k$ such that: 
\begin{enumerate}[\hspace{0.7cm}(a)$_k$]
\item the binary digits of $X$ at all positions in the interval 
$I_k$ are  1.
\item some digit of $X$ in the interval $J_k$ is zero.
\end{enumerate}


The properties (a)$_k$ and (b)$_k$ are effective, in the sense that the set of reals satisfying
them is a finite union of basic open sets, which are uniformly computable in $k$.
Note that,  since  $t_k+g(t_k)<t_{k+1}$ for all sufficiently large $k$, 
 the properties (a)$_k$ are independent\footnote{as events in the probability space of all reals, where
 the event corresponding to the property is the set of all reals that
satisfy this property.} 
 for all sufficiently large $k$. Since $g$ is nondecreasing, the same is true of the properties (b)$_k$.
The measure of reals that meet property
(a)$_k$ is $2^{-\lfloor g(t_k) \rfloor -1}$. Also, the measure of reals that do not meet property (b)$_k$
is $2^{t_k-t_{k+1}}$. Hence, by the effective Borel-Cantelli lemma: 
\begin{equation}
\parbox{13cm}{For any \ml random real there exist infinitely many $k$ such that
(a)$_k$ holds and finitely many $k$ such that (b)$_k$ does not hold.}
\end{equation}

Now, given $\Omega$, let $c$ be a number such that for all $k\geq c$, $t_k+g(t_k)<t_{k+1}$ and 
the property (b)$_k$ is met by $\Omega$.   Let $\beta$ be defined as in (\ref{HfHfPGgk7}) for that value of $c$. Define $d_k=t_k+\lfloor g(t_k) \rfloor +1$.
Let us say that {\em $k$ is valid} if
it is larger than $c$ and (a)$_k$ holds for $\Omega$. 
Then the following holds:  
\begin{equation}\label{eq:reducpfrop}
\parbox{13cm}{If $k$ is valid then $\Omega(d_k)=1\iff\beta(t_k)=0$.}
\end{equation}
In order to see this, note first that satisfaction of $(b)_{k'}$ for all $k'\geq k$ (where $k>c$) 
means that $\beta$ agrees with $\Omega+2^{-d_k}$ on all digits in the interval 
\[
\Big[t_{k-1}+\lfloor g(t_{k-1})\rfloor +2, t_{k}+\lfloor g(t_k)\rfloor+1\Big). 
\]
First suppose that $\Omega(d_k)=0$. In this case, adding $2^{-d_k}$ to $\Omega$ (as one term in the sum $\sum_{i>c} 2^{-(t_i+\lfloor g(t_i) \rfloor+1)}$) will cause it to change at position $d_k$ but leave it unchanged at position $t_k$, meaning that  $\beta(t_k)=1$. Suppose, on the other hand that  
$\Omega(d_k)=1$. Let $j$ be the greatest $\leq d_k$ such that $\Omega(j)=0$, so that $j\in [t_{k-1}+\lfloor g(t_{k-1}) \rfloor+2, t_{k}-1]$ because $(a)_k$ is satisfied as well as $(b)_{k-1}$. The addition of $2^{-d_k}$ to $\Omega$ will cause the digit at position $j$ to become 1, while making the digit at position $t_k$ into a 0. Thus $\beta(t_k)=0$ in this case.

If $\Omega$ computes $\beta$ with redundancy $g$, then 
for each $k$, computing $\beta(t_k)$ uses at most  the first
$\lfloor t_k+g(t_k)\rfloor $ bits of $\Omega$. 
Then \eqref{eq:reducpfrop} establishes that for the special case where $k$ is valid,
the first  $\lfloor t_k+g(t_k)\rfloor $ bits of $\Omega$ are enough to decide  
$\Omega(t_k+\lfloor g(t_k) \rfloor+1)$. This shows that there is a partial computable
prediction rule for the digits of $\Omega$. In other words, there is a c.e.\  martingale
that succeeds on $\Omega$,  contradicting the fact that $\Omega$ is \ml random.
\end{proof}
It remains to show that a sequence $(t_i)$ as in the hypothesis of
Lemma \ref{23akY6Ldyu} exists for the function $\log (n)$. 

\begin{lem}[Existence of partition]\label{b4eX9dfKaQ}
If $g(n)=\log (n)$,  there exists a computable increasing sequence
$(t_i)$ such that $t_i +g(t_i)<t_{i+1}$ for all sufficiently large $i$ and such that \eqref{SJlX3zHJot} holds.
\end{lem}
\begin{proof}
For $k\geq 1$ define:  \[ t_k=  \sum_{i=1}^k (\log (i) +2\log\log (i)). \]  
Then $t_{k+1}-t_k=\log (k+1) +2\log\log (k+1)$ so the
second
clause of \eqref{SJlX3zHJot} holds by Lemma \ref{QaQK5kYLLF}. Next we show that for all sufficiently large $k$, $t_k+g(t_k)<t_{k+1}$. 
Since  $\log (i) +2\cdot\log\log (i) \leq 2 \cdot\log (i)$ we have: 
\[
t_k\leq 2 \cdot \sum_{i\leq k} \log (i)
= 2\cdot\log \Big((k)!\Big)
\leq 2k\cdot\log (k).
\]
 Hence:
\begin{equation}\label{wimtOBNGQI}
g(t_k)\leq g(2k\cdot \log (k))=1+\log (k) +\log\log (k).  
\end{equation}
For all sufficiently large $k$ the last expression is bounded above by $\log (k+1) +2\cdot\log\log (k+1)$.
Hence $g(t_k)< t_{k+1}-t_k$ for all sufficiently large $k$, as promised.  Moreover, by \eqref{wimtOBNGQI} and Lemma \ref{QaQK5kYLLF},
the first clause of
\eqref{SJlX3zHJot} holds for the sequence $(t_i)$, which concludes the proof of this lemma.
\end{proof}

These lemmas conclude the proof of Theorem  \ref{uVwUvgKQIJ}.

\subsection{Computable functions as upper bounds of oracle use of omega}\label{B1WbTtcepf}
Theorem \ref{uVwUvgKQIJ} gives a rather precise picture of the rate of growth
of the oracle use in computations of one halting probability from another. 
It is reasonable to ask if we can obtain a more general characterisation of the 
upper bounds on the oracle use in such computations.
Lemma \ref{F4dKJufHui}, for example, suggests that this might be possible. Any computable function
$g$ such that $\sum_i 2^{-g(i)}$ converges is such an upper bound. Could this condition characterise
these upper bounds? In other words, is it true that if 
$\sum_i 2^{-g(i)}$ diverges then
there are two omega numbers such that one is not computable from the other with redundancy $g$?
If for all computable $g$ such that $\sum_i 2^{-g(i)}$ diverges, there existed a sequence $(t_i)$  satisfying the conditions described in  Lemma \ref{23akY6Ldyu},  a positive answer would follow.
The following proposition establishes that this is not the case, even if $g$ is assumed to be monotone.

\begin{prop}\label{vWF7Tcg5Pj}
There exists a nondecreasing computable function $g$ such that $\sum_i 2^{-g(i)}=\infty$ and the following holds for all increasing sequences of positive integers $(t_i)$: 
\begin{equation}\label{ueUtC1kCTt}
\sum_i 2^{-g(t_i)}=\infty\Rightarrow
\sum_i 2^{t_i-t_{i+1}}=\infty.
\end{equation}
\end{prop}

For the proof of Proposition \ref{vWF7Tcg5Pj}, 
we construct $g$ of the form $\mbox{log}(f(n))$. Let $m_i=2^{2^i+i}$
and  define $f(n)=m_1$ for $m_1$ 
many values of $n$ (i.e.\ the interval $[1,m_1]$), 
then $f(n)=m_2$ for the next $m_2$ values of $n$, and so on.
This determines $f$ and $g$, so it remains to show that $g$ has the desired property.
Define the intervals $I_n=\{i\ |\ f(i)=m_n\}$ and note that by the definition of $f$ we have
$|I_n|=m_n$.
Moreover, these intervals are a partition of $\Nat$.
In order to understand the intuition behind the definition of $f$, imagine momentarily that the sequence $(m_i)$ has not yet been specified,  and let $f$ be defined as above with respect to some sequence $(m_i)$ which is yet to be determined. Through various considerations we shall arrive at the specific
definition of $(m_i)$ given previously, along with a proof that this particular choice confers the desired properties on $f$.
Since 
$g(n)=\mbox{log}(f(n))$, if $f(n)=m$, this will contribute $1/m$ to the sum $\sum_n 2^{-g(n)}$. 
So those values of $n$ for which $f(n)=m_1$ contribute (in total) 1 to the sum $\sum_n 2^{-g(n)}$. 
Then those values of $n$ for which $f(n)=m_2$ contribute the same amount again, and so on. 
This ensures that $\sum_n 2^{-g(n)}$ is infinite, for any  choice of $(m_i)$. 
So it remains to specify $(m_i)$ so that \eqref{ueUtC1kCTt} is met.

Given any increasing sequence $(t_n)$ such that 
$\sum_n 2^{-g(t_n)}=\infty$, we wish to show that $\sum_n 2^{t_n-t_{n+1}} =\infty$. 
For each $n$, consider the set $J_n(t)=|\{i\ |\ t_i\in I_n\}|$,
(where the suffix $(t)$ indicates the dependence on $(t_i)$). Since $|I_n|=m_n$ we have:
\[
\sum_i 2^{-g(t_i)}=\sum_i \frac{1}{f(t_i)}=
\sum_n \frac{|J_n(t)|}{m_n}=\sum_n \frac{|J_n(t)|}{|I_n|}.
\]
Since we are given that 
  $\sum_i 2^{-g(t_i)}=\infty$, it follows that there are infinitely many $n$ with
$|J_n(t)|> |I_n|\cdot 2^{-n}$. 
Let $D(t)$ be the set of such $n$, 
where once again the suffix $(t)$ indicates the dependence on $(t_i)$,
\ie $D(t)=\{n\ |\ |J_n(t)|> |I_n|\cdot 2^{-n}\}$.
We aim to show that, so long as we specify the sequence $(m_i)$ appropriately: 
\begin{equation} \label{ret} 
\textrm{if \hspace{0.4cm}$n\in D(t)$ \hspace{0.4cm}then}
\hspace{0.4cm}
\sum_{i\in J_n(t)} 2^{t_i-t_{i+1}} \geq 2^{-2}.
\end{equation} 
Since
$\sum_i 2^{-g(t_i)}=\infty$
implies that
$D(t)$ is infinite,
this suffices to give the result because:
\begin{equation}\label{2XM9dQOfhx}
\sum_{i} 2^{t_i-t_{i+1}} \geq 
\sum_{n\in D(t)}\left(\sum_{i\in J_n(t)} 2^{t_i-t_{i+1}}\right).
\end{equation}
The rough idea is that, 
for an adversary who wishes to keep the sum in \eqref{ret} small, 
the optimal approach is to 
ensure that the terms $t_i$, $i\in J_n(t)$
are spaced
evenly in $I_n$. This rough idea is formalised in the following lemma.

\begin{lem}[Minimizing the sums]\label{MHKEk5bFl3}
Let $I$ be an interval, $k>1$, and suppose that $t_i$, $1\leq i\leq k+1$ is an increasing sequence of
numbers in $I$. Then $\sum_{1\leq i\leq k} 2^{-(t_{i+1}-t_i)}\geq k\cdot 2^{-\ceil{|I|/k}}\geq k\cdot 2^{-|I|/k-1}$.
\end{lem}
\begin{proof}
The second inequality above is obvious, so we only need to prove the first one. 
Let us consider the $t_i$ as movable markers. We can assume that the first marker $t_1$ is placed on the first element of the interval $I$, and that $t_{k+1}$ is placed on the last element of the interval, since otherwise they can be moved there reducing the sum   $\sum_{1\leq i\leq k} 2^{-(t_{i+1}-t_i)}$. 
Let $c_i=t_{i+1}-t_i$ and let $m=t_{k+1}-t_1$, so that $\sum_{i=1}^k c_i =m$. Consider $\Pi_k^m$, which is the set of all sequences $a_1,a_2,\dots, a_k$ such that every $a_i\in \mathbb{N}^+$, and  $\sum_{i=1}^k a_i =m $.  The result follows if we can establish that amongst all sequences in $\Pi_k^m$, the minimum value of $\sum_i 2^{-a_i}$ is attained whenever $a_i \in \{ \lfloor  m/k \rfloor,  \lceil  m/k \rceil \}$ for all $i$.  We show this by induction on $m\geq k$. For $m=k$ the result follows easily, since we must have $a_i=1$ for all $i$. Suppose the result holds for $m$. Let $\hat \Pi_k^m$ be all those sequences in $\Pi_k^m$  such that $a_i \in \{ \lfloor  m/k \rfloor,  \lceil  m/k \rceil \}$ for all $i$. It is clear that for all sequences in $\hat \Pi_k^m$ the sum   $\sum_i 2^{-a_i}$ is the same (and similarly for $\hat \Pi_k^{m+1}$). Now suppose there exist two sequences $(a_i)$ and $(b_i)$ in $\Pi_k^{m+1}$ such that  $(a_i)\in \hat \Pi_k^{m+1}$,   $(b_i)\notin \hat \Pi_k^{m+1}$ and for which $\sum_i 2^{-b_i} <\sum_i 2^{-a_i}$. Let $c$ be such that $b_c$ is largest, and let $d$ be such that $a_d$ is largest. Then we have $b_c\geq a_d$. Let $(a_i^{\ast})$ be the element of $\hat \Pi_k^m$ obtained by replacing $a_d$ by $a_d-1$. Let $(b_i^{\ast})$ be the element of $\Pi_k^m$ obtained by replacing $b_c$ by $b_c-1$.
 Then we have  $\sum_i 2^{-b_i^{\ast}} <\sum_i 2^{-a_i^{\ast}}$, contradicting the induction hypothesis.
\end{proof} 

Since $|I_n|=m_n$ and $|J_n(t)|\geq 2^{-n}\cdot |I_n|$ for each $n\in D(t)$,
by Lemma \ref{MHKEk5bFl3} we have that for each $n\in D(t)$,
\begin{equation}\label{FjlB39Kcrp}
\sum_{i\in J_n(t)} 2^{t_i-t_{i+1}} \geq
(|J_n|-1)\cdot 2^{-|I_n|/|J_n(t)|-1}\geq
(2^{-n}\cdot m_n-1) \cdot 2^{-2^{n}-1}\geq
(m_n-1)\cdot 2^{-2^{n}-n-1}.
\end{equation}
So if we define $m_n=2^{2^n+n}$ (as previously), 
we get the required inequality \eqref{ret}. 
We summarize the argument. We already noted that for 
any choice of $m_i$, the corresponding function $g$
satisfies $\sum_n 2^{-g(n)}=\infty$. Now fix $m_i=2^{2^i+i}$ and assume that 
$\sum_i 2^{-g(t_i)}=\infty$ for some increasing sequence $(t_i)$.
 By Lemma \ref{MHKEk5bFl3}
 we get \eqref{FjlB39Kcrp}. So by \eqref{2XM9dQOfhx}
 we have that $\sum_{i} 2^{t_i-t_{i+1}}=\infty$ which concludes the proof.
 
\section{Concluding remarks}
We have characterised the redundancy growth rate which is generally required in computations of 
halting probabilities from other halting probabilities, in terms of the functions $h_{\epsilon}(n)$.  
It would be pleasing to obtain a more general characterisation of the required redundancy in such 
computations, in terms of arbitrary computable nondecreasing functions $g$ such that $\sum_n 2^{-g(n)}$
converges. 
Although our analysis applies to this generalised goal with respect to the upper bounds
that we obtain (Lemma \ref{F4dKJufHui}), Proposition \ref{vWF7Tcg5Pj}
indicates that our lower bound analysis (Lemma \ref{23akY6Ldyu}) 
may not be sufficient for such a generalisation.

\end{document}